\documentclass{llncs}
\usepackage{amsmath,amssymb,amsrefs,mathrsfs,lineno,fancyvrb,tikz,wrapfig}
\usepackage{hyperref}
\usetikzlibrary{arrows,automata,shapes}

\font\cyreight=wncyr8

\newcommand\N{{\mathbb N}}
\newcommand\Z{{\mathbb Z}}
\newcommand\R{{\mathbb R}}

\DeclareMathOperator\sym{Sym}

\newtheorem{algorithm}{Algorithm}
\newtheorem{assumption}{Standing assumption}

\newenvironment{fsa}[1][auto]{\begin{tikzpicture}[->,>=stealth',
    shorten >=1pt,auto,node distance=3cm,double distance between line centers=0.45ex,
    initial text=,accepting/.style=accepting by arrow,
    every state/.style={inner sep=3pt,minimum size=0pt},
    every loop/.style={looseness=12},semithick,#1]}{\end{tikzpicture}}

\begin{document}
\title{Algorithmic decidability of Engel's property for automaton groups}
\author{Laurent Bartholdi\thanks{Partially supported by ANR grant ANR-14-ACHN-0018-01}\inst1\inst2}
\institute{\'Ecole Normale Sup\'erieure, Paris. \email{laurent.bartholdi@ens.fr}\and Georg-August-Universit\"at zu G\"ottingen. \email{laurent.bartholdi@gmail.com}}
\maketitle
\begin{abstract}
  We consider decidability problems associated with Engel's identity
  ($[\cdots[[x,y],y],\dots,y]=1$ for a long enough commutator
  sequence) in groups generated by an automaton.

  We give a partial algorithm that decides, given $x,y$, whether an
  Engel identity is satisfied. It succeeds, importantly, in proving
  that Grigorchuk's $2$-group is not Engel.

  We consider next the problem of recognizing Engel elements, namely
  elements $y$ such that the map $x\mapsto[x,y]$ attracts to $\{1\}$.
  Although this problem seems intractable in general, we prove that it
  is decidable for Grigorchuk's group: Engel elements are precisely
  those of order at most $2$.

  Our computations were implemented using the package \textsc{Fr}
  within the computer algebra system \textsc{Gap}.
\end{abstract}

%%%%%%%%%%%%%%%%%%%%%%%%%%%%%%%%%%%%%%%%%%%%%%%%%%%%%%%%%%%%%%%%
\section{Introduction}
A \emph{law} in a group $G$ is a word $w=w(x_1,x_2,\dots,x_n)$ such
that $w(g_1,\dots,g_n)=1$ for all $g_1,\dots,g_n\in G$; for example,
commutative groups satisfy the law
$[x_1,x_2]=x_1^{-1}x_2^{-1}x_1x_2$. A \emph{variety} of groups is a
maximal class of groups satisfying a given law; e.g.\ the variety of
commutative groups (satisfying $[x_1,x_2]$) or of groups of exponent
$p$ (satisfying $x_1^p$).

Consider now a sequence $\mathscr W=(w_0,w_1,\dots)$ of words in $n$
letters. Say that $(g_1,\dots,g_n)$ \emph{almost satisfies}
$\mathscr W$ if $w_i(g_1,\dots,g_n)=1$ for all $i$ large enough, and
say that $G$ \emph{almost satisfies $\mathscr W$} if all $n$-tuples
from $G$ almost satisfy $\mathscr W$. For example, $G$ almost
satisfies $(x_1,\dots,x_1^{i!},\dots)$ if and only if $G$ is a torsion
group.

\begin{wrapfigure}[9]{r}{55mm}
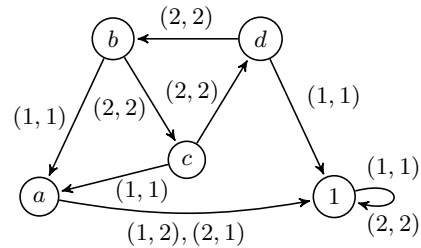

\small\vspace{-7ex}\begin{fsa}[scale=0.7]
  \node[state] (b) at (1.4,3) {$b$};
  \node[state] (d) at (4.2,3) {$d$};
  \node[state] (c) at (2.8,0.7) {$c$};
  \node[state] (a) at (0,0) {$a$};
  \node[state] (e) at (5.6,0) {$1$};
  \path (b) edge node[left,pos=0.6] {$(2,2)$} (c) edge node[left] {$(1,1)$} (a)
        (c) edge node[left,pos=0.6] {$(2,2)$} (d) edge node[below right=-1mm] {$(1,1)$} (a)
        (d) edge node[above] {$(2,2)$} (b) edge node {$(1,1)$} (e)
        (a) edge[bend right=10] node[below] {$(1,2),(2,1)$} (e)
        (e) edge[loop right] node[above=1mm] {$(1,1)$} node[below=1mm] {$(2,2)$} (e);
\end{fsa}\vspace{-2ex}
\caption{The Grigorchuk automaton}\label{fig:grigorchuk}
\end{wrapfigure}

The problem of deciding algorithmically whether a group belongs to a
given variety has received much attention~(see
e.g.~\cite{jackson:undecidable} and references therein); we consider
here the harder problems of determining whether a group (respectively
a tuple) almost satisfies a given sequence. This has, up to now, been
investigated mainly for the torsion sequence
above~\cite{godin-klimann-picantin:torsion-free}.

The first Grigorchuk group $G_0$ is an automaton group which appeared
prominently in group theory, for example as a finitely generated
infinite torsion group~\cite{grigorchuk:burnside} and as a group of
intermediate word growth~\cite{grigorchuk:growth};
see~\S\ref{ss:grigorchuk}. It is the group of automatic
transformations of $\{1,2\}^\infty$ generated by the five states of
the automaton from Figure~\ref{fig:grigorchuk}, with input and output
written as $(\textsf{in},\textsf{out})$.

The Engel law is
\[E_c=E_c(x,y)=[x,y,\dots,y]=[\cdots[[x,y],y],\dots,y]\]
with $c$ copies of `$y$'; so $E_0(x,y)=x$, $E_1(x,y)=[x,y]$ and
$E_c(x,y)=[E_{c-1}(x,y),y]$. See below for a motivation. Let us call a
group (respectively a pair of elements) \emph{Engel} if it almost
satisfies $\mathscr E=(E_0,E_1,\dots)$.  Furthermore, let us call
$h\in G$ an \emph{Engel} element if $(g,h)$ is Engel for all $g\in G$.

\noindent A concrete consequence of our investigations is:
\begin{theorem}\label{thm:1}
  The first Grigorchuk group $G_0$ is not Engel. Furthermore, an
  element $h\in G_0$ is Engel if and only if $h^2=1$.
\end{theorem}
We prove a similar statement for another prominent example of automaton
group, the \emph{Gupta-Sidki group}, see Theorem~\ref{thm:gs}.

Theorem~\ref{thm:1} follows from a partial algorithm, giving a
criterion for an element $y$ to be Engel. This algorithm proves, in
fact, that the element $ad$ in the Grigorchuk group is not Engel. We
consider the following restricted situation, which is general as far
as the Engel property is concerned, see~\S\ref{ss:automata}: an
\emph{automaton group} is a group $G$ endowed with extra data, in
particular with a family of self-maps called \emph{states}, indexed by
a set $X$ and written $g\mapsto g@x$ for $x\in X$; it is
\emph{contracting} for the word metric $\|\cdot\|$ on $G$ if there are
constants $\eta<1$ and $C$ such that $\|g@x\|\le\eta\|g\|+C$ holds for
all $g\in G$ and all $x\in X$.  Our aim is to solve the following
\textbf{decision problems} in an automaton group $G$:
\begin{description}
\item[Engel($g,h$)] Given $g,h\in G$, does there exist $c\in\N$ with
  $E_c(g,h)$?
\item[Engel($h$)] Given $h\in G$, does Engel($g,h$) hold for all $g\in G$?
\end{description}

\noindent The algorithm is described in~\S\ref{ss:algo}. As a consequence,
\begin{corollary}
  Let $G$ be an automaton group acting on the set of binary sequences
  $\{1,2\}^*$, that is contracting with contraction coefficient
  $\eta<1$. Then, for torsion elements $h$ of order $2^e$ with
  $2^{2^e}\eta<1$, the property Engel($h$) is decidable.
\end{corollary}

The Engel property attracted attention for its relation to nilpotency:
indeed a nilpotent group of class $c$ satisfies $E_c$, and conversely
among compact~\cite{medvedev:engel} and
solvable~\cite{gruenberg:engel} groups, if a group satisfies $E_c$ for
some $c$ then it is locally nilpotent. Conjecturally, there are
non-locally nilpotent groups satisfying $E_c$ for some $c$, but this
is still unknown. It is also an example of iterated identity,
see~\cites{erschler:iteratedidentities,bandman+:dynamics}. In
particular, the main result of~\cite{bandman+:dynamics} implies easily
that the Engel property is decidable in algebraic groups.

It is comparatively easy to prove that the first Grigorchuk group $G_0$
satisfies no law~\cites{abert:nonfree,leonov:identities}; this result
holds for a large class of automaton groups. In fact, if a group
satisfies a law, then so does its profinite completion. In the class
mentioned above, the profinite completion contains abstract free
subgroups, precluding the existence of a law. No such arguments would
help for the Engel property: the restricted product of all finite
nilpotent groups is Engel, but the unrestricted product again contains
free subgroups. This is one of the difficulties in dealing with
iterated identities rather than identities.

If $\mathfrak A$ is a nil algebra (namely, for every $a\in\mathfrak A$
there exists $n\in\N$ with $a^n=0$) then the set of elements of the
form $\{1+a:a\in\mathfrak A\}$ forms a group $1+\mathfrak A$ under the
law $(1+a)(1+b)=1+(a+b+ab)$. If $\mathfrak A$ is defined over a field
of characteristic $p$, then $1+\mathfrak A$ is a torsion group since
$(1+a)^{p^n}=1$ if $a^{p^n}=0$. Golod constructed
in~\cite{golod:burnside} non-nilpotent nil algebras $\mathfrak A$ all
of whose $2$-generated subalgebras are nilpotent (namely,
$\mathfrak A^n=0$ for some $n\in\N$); given such an $\mathfrak A$, the
group $1+\mathfrak A$ is Engel but not locally nilpotent.

Golod introduced these algebras as means of obtaining infinite,
finitely generated, residually finite (every non-trivial element in
the group has a non-trivial image in some finite quotient), torsion
groups. Golod's construction is highly non-explicit, in contrast with
Grigorchuk's group for which much can be derived from the automaton's
properties.

It is therefore of high interest to find explicit examples of Engel
groups that are not locally nilpotent, and the methods and algorithms
presented here are a step in this direction.

An important feature of automaton groups is their amenability to
computer experiments, and even as in this case of rigorous
verification of mathematical assertions;
see also~\cite{klimann-mairesse-picantin:computations}, and the numerous
decidability and undecidability of the finiteness property
in~\cites{akhavi-klimann-lombardy-mairesse-picantin:finiteness,gillibert:finiteness,klimann:finiteness}.

The proof of Theorem~\ref{thm:1} relies on a computer calculation. It could
be checked by hand, at the cost of quite unrewarding effort. One of the
purposes of this article is, precisely, to promote the use of computers
in solving general questions in group theory: the calculations performed,
and the computer search involved, are easy from the point of view of a
computer but intractable from the point of view of a human.

The calculations were performed using the author's group theory
package \textsc{Fr}, specially written to manipulate automaton
groups. This package integrates with the computer algebra system
\textsc{Gap}~\cite{gap4:manual}, and is freely available from the
\textsc{Gap} distribution site
\[\verb+http://www.gap-system.org+\]

%%%%%%%%%%%%%%%%%%%%%%%%%%%%%%%%%%%%%%%%%%%%%%%%%%%%%%%%%%%%%%%%
\section{Automaton groups}\label{ss:automata}
An \emph{automaton group} is a finitely generated group associated with a
Mealy automaton. We define a \emph{Mealy automaton} $\mathscr M$ as a graph
such as that in Figure~\ref{fig:grigorchuk}. It has a set of \emph{states}
$Q$ and an \emph{alphabet} $X$, and there are \emph{transitions} between
states with \emph{input} and \emph{output} labels in $X$, with the
condition that, at every state, all labels appear exactly once as input and
once as output on the set of outgoing transitions.

Every state $q\in Q$ of $\mathscr M$ defines a transformation, written as
exponentiation by $q$, of the set of words $X^*$ by the following rule:
given a word $x_1\dots x_n\in X^*$, there exists a unique path in the
automaton starting at $q$ and with input labels $x_1,\dots,x_n$. Let
$y_1,\dots,y_n$ be its corresponding output labels. Then declare $(x_1\dots
x_n)^q=y_1\dots y_n$.

The action may also be defined recursively as follows: if there is a
transition from $q\in Q$ to $r\in Q$ with input label $x_1\in X$ and
output label $y_1\in X$, then
$(x_1 x_2\dots x_n)^q=y_1\,(x_2\dots,x_n)^r$.

By the \emph{automaton group} associated with the automaton $\mathscr M$,
we mean the group $G$ of transformations of $X^*$ generated by $\mathscr
M$'s states. Note that all elements of $G$ admit descriptions by automata;
namely, a word of length $n$ in $G$'s generators is the transformation
associated with a state in the $n$-fold product of the automaton of
$G$. See~\cite{gecseg-c:ata} for the abstract theory of automata,
and~\cite{gecseg:products} for products more specifically.

The structure of the automaton $\mathscr M$ may be encoded in an injective
group homomorphism $\psi\colon G\to G^X\rtimes\sym(X)$ from $G$ to the
group of \emph{$G$-decorated permutations of $X$}. This last group --- the
\emph{wreath product} of $G$ with $\sym(X)$ --- is the group of
permutations of $X$, with labels in $G$ on each arrow of the permutation;
the labels multiply as the permutations are composed. The construction of
$\psi$ is as follows: consider $q\in Q$. For each $x\in X$, let $q@x$
denote the endpoint of the transition starting at $q$ with input label $x$,
and let $x^\pi$ denote the output label of the same transition; thus every
transition in the Mealy automaton gives rise to
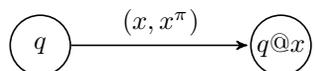
\[\begin{fsa}[scale=0.8]
  \node[state,minimum size=8mm] (q) at (0,0) {$q$};
  \node[state,inner sep=0pt,minimum size=8mm] (r) at (4,0) {$q@x$};
  \path (q) edge node[above] {$(x,x^\pi)$} (r);
\end{fsa}.
\]
The transformation $\pi$ is a permutation of $X$, and we set
\[\psi(q)=\langle x\mapsto q@x\rangle\pi,\]
namely the permutation $\varepsilon$ with decoration $q@x$ on the
arrow from $x$ to $x^\pi$.

We generalize the notation $q@x$ to arbitrary words and group
elements. Consider a word $v\in X^*$ and an element $g\in G$; denote
by $v^g$ the image of $v$ under $g$. There is then a unique element of
$G$, written $g@v$, with the property
\[(v\,w)^g=(v^g)\,(w)^{g@v}\text{ for all }w\in X^*.\]
We call by extension this element $g@v$ the \emph{state} of $g$ at
$v$; it is the state, in the Mealy automaton defining $g$, that is
reached from $g$ by following the path $v$ as input; thus in the
Grigorchuk automaton $b@1=a$ and $b@222=b$. There is a reverse
construction: by $v*g$ we denote the permutation of $X^*$ defined by
\[(v\,w)^{v*g}=v\,w^g,\qquad w^{v*g}=w\text{ if $w$ does not start with $v$}.
\]
Given a word $w=w_1\dots w_n\in X^*$ and a Mealy automaton $\mathscr M$ of
which $g$ is a state, it is easy to construct a Mealy automaton of which
$w*g$ is a state: add a path of length $n$ to $\mathscr M$, with input and
output $(w_1,w_2),\dots,(w_n,w_n)$ along the path, and ending at
$g$. Complete the automaton with transitions to the identity element. Then
the first vertex of the path defines the transformation $w*g$. For example,
here is $12*d$ in the Grigorchuk automaton:
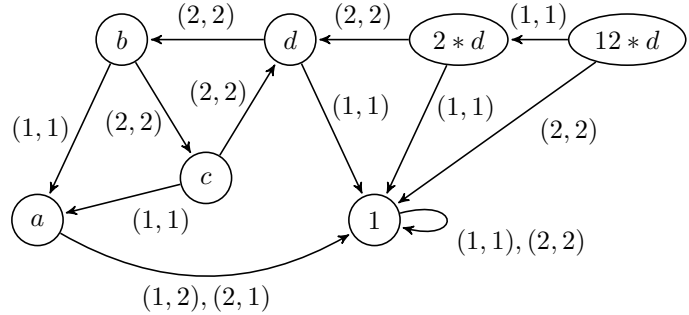
\[\begin{fsa}[scale=0.8,every state/.style={minimum size=4.5ex}]
  \node[state] (b) at (1.4,3) {$b$};
  \node[state] (d) at (4.2,3) {$d$};
  \node[state,ellipse] (2d) at (7.0,3) {$2*d$};
  \node[state,ellipse] (12d) at (9.8,3) {$12*d$};
  \node[state] (c) at (2.8,0.7) {$c$};
  \node[state] (a) at (0,0) {$a$};
  \node[state] (e) at (5.6,0) {$1$};
  \path (b) edge node[left,pos=0.66] {$(2,2)$} (c) edge node[left] {$(1,1)$} (a)
        (c) edge node[left,pos=0.66] {$(2,2)$} (d) edge node {$(1,1)$} (a)
        (d) edge node[above] {$(2,2)$} (b) edge node[right,pos=0.33] {$(1,1)$} (e)
        (a) edge[bend right=30] node[below] {$(1,2),(2,1)$} (e)
        (e) edge[loop right] node[below right] {$(1,1),(2,2)$} (e)
        (12d) edge node[above] {$(1,1)$} (2d) edge node [pos=0.33] {$(2,2)$} (e)
        (2d) edge node[above] {$(2,2)$} (d) edge node [right,pos=0.33] {$(1,1)$} (e);
      \end{fsa}
\]
Note the simple identities $(g@v_1)@v_2=g@(v_1v_2)$,
$(v_1v_2)*g=v_1*(v_2*g)$, and $(v*g)@v=g$. Recall that we write
conjugation in $G$ as $g^h=h^{-1}gh$. For any $h\in G$ we have
\begin{equation}\label{eq:twist}
  (v*g)^h=v^h*(g^{h@v}).
\end{equation}

An automaton group is called \emph{regular weakly branched} if there
exists a non-trivial subgroup $K$ of $G$ such that $\psi(K)$ contains
$K^X$. In other words, for every $k\in K$ and every $v\in X^*$, the
element $v*k$ also belongs to $K$, and therefore to $G$. Ab\'ert
proved in~\cite{abert:nonfree} that regular weakly branched groups
satisfy no law.

In this text, we concentrate on the Engel property, which is
equivalent to nilpotency for finite groups. In particular, if an
automaton group $G$ is to have a chance of being Engel, then its image
under the map $G\to G^X\rtimes\sym(X)\to\sym(X)$ should be a nilpotent
subgroup of $\sym(X)$. Since finite nilpotent groups are direct
products of their $p$-Sylow subgroups, we may reduce to the case in
which the image of $G$ in $\sym(X)$ is a $p$-group. A further
reduction lets us assume that the image of $G$ is an abelian subgroup
of $\sym(X)$ of prime order. We therefore make the following
\begin{assumption}
  The alphabet is $X=\{1,\dots,p\}$, and automaton groups are defined by
  embeddings $\psi\colon G\to G^p\rtimes\Z_{/p}$, with $\Z_{/p}$ the cyclic
  subgroup of $\sym(X)$ generated by the cycle $(1,2,\dots,p)$.
\end{assumption}
\noindent This is the situation considered in the Introduction.

We make a further reduction in that we only consider the Engel
property for elements of finite order. This is not a very strong
restriction: given $h$ of infinite order, one can usually find an
element $g\in G$ such that the conjugates $\{g^{h^n}:n\in\Z\}$ are
independent, and it then follows that $h$ is not Engel. This will be
part of a later article.

\subsection{Grigorchuk's first group}\label{ss:grigorchuk}
This section is not an introduction to Grigorchuk's first group, but
rather a brief description of it with all information vital for the
calculation in~\S\ref{ss:proof}. For more details, see
e.g.~\cite{bartholdi-g-s:bg}.

Fix the alphabet $X=\{1,2\}$. The first Grigorchuk group $G_0$ is a
permutation group of the set of words $X^*$, generated by the four
non-trivial states $a,b,c,d$ of the automaton given in
Figure~\ref{fig:grigorchuk}. Alternatively, the transformations
$a,b,c,d$ may be defined recursively as follows:
\begin{equation}\label{eq:grigorchuk}
\begin{alignedat}{2}
  (1x_2\dots x_n)^a&=2x_2\dots x_n, &\qquad (2x_2\dots x_n)^a&=1x_2\dots x_n,\\
  (1x_2\dots x_n)^b&=1a(x_2\dots x_n), &\qquad (2x_2\dots x_n)^b&=2c(x_2\dots x_n),\\
  (1x_2\dots x_n)^c&=1a(x_2\dots x_n), &\qquad (2x_2\dots x_n)^c&=2d(x_2\dots x_n),\\
  (1x_2\dots x_n)^d&=1x_2\dots x_n, &\qquad (2x_2\dots x_n)^d&=2b(x_2\dots x_n)
\end{alignedat}
\end{equation}
which directly follow from $d@1=1$, $d@2=b$, etc.

It is remarkable that most properties of $G_0$ derive from a careful
study of the automaton (or equivalently this action), usually using
inductive arguments. For example,
\begin{proposition}[\cite{grigorchuk:burnside}]\label{prop:torsion}
  The group $G_0$ is infinite, and all its elements have order a power of $2$.
\end{proposition}

\noindent The self-similar nature of $G_0$ is made apparent in the following manner:
\begin{proposition}[\cite{bartholdi-g:parabolic}*{\S4}]\label{prop:branch}
  Define $x=[a,b]$ and $K=\langle x,x^{c},x^{ca}\rangle$. Then $K$ is
  a normal subgroup of $G_0$ of index $16$, and $\psi(K)$ contains
  $K\times K$.

  In other words, for every $g\in K$ and every $v\in X^*$ the element
  $v*g$ belongs to $G_0$.
\end{proposition}

%%%%%%%%%%%%%%%%%%%%%%%%%%%%%%%%%%%%%%%%%%%%%%%%%%%%%%%%%%%%%%%%
\section{A semi-algorithm for deciding the Engel property}\label{ss:algo}

We start by describing a semi-algorithm to check the Engel property. It
will sometimes not return any answer, but when it returns an answer then that
answer is guaranteed correct. It is guaranteed to terminate as long as the
contraction property of the automaton group $G$ is strong enough.

\begin{algorithm}\label{algo:1}
  Let $G$ be a contracting automaton group with alphabet
  $X=\{1,\dots,p\}$ for prime $p$, with the contraction property
  $\|g@j\|\le\eta\|g\|+C$.

  For $n\in p\N$ and $R\in\R$ consider the following finite graph
  $\Gamma_{n,R}$. Its vertex set is $B(R)^n\cup\{\texttt{fail}\}$,
  where $B(R)$ denotes the set of elements of $G$ of length at most
  $R$. Its edge set is defined as follows: consider a vertex
  $(g_1,\dots,g_n)$ in $\Gamma_{n,R}$, and compute
  \[(h_1,\dots,h_n)=(g_1^{-1}g_2,\dots,g_n^{-1}g_1).\]
  If $h_i$ fixes $X$ for all $i$, i.e.\ all $h_i$ have trivial image
  in $\sym(X)$, then for all $j\in\{1,\dots,p\}$ there is an edge from
  $(g_1,\dots,g_n)$ to $(h_1@j,\dots,h_n@j)$, or to \texttt{fail} if
  $(h_1@j,\dots,h_n@j)\not\in B(R)^n$. If some $h_i$ does not fix $X$,
  then there is an edge from $(g_1,\dots,g_n)$ to $(h_1,\dots,h_n)$,
  or to \texttt{fail} if $(h_1,\dots,h_n)\not\in B(R)^n$.
  \begin{description}
  \item[\boldmath Given $g,h\in G$ with $h^n=1$:] Set
    $t_0=(g,g^h,g^{h^2},\dots,g^{h^{n-1}})$. If there exists $R\in\N$
    such that no path in $\Gamma_{n,R}$ starting at $t_0$ reaches
    \texttt{fail}, then Engel($g,h$) holds if and only if the only
    cycle in $\Gamma_{n,R}$ reachable from passes through
    $(1,\dots,1)$.

    If the contraction coefficient satisfies $2^n\eta<1$, then it is
    sufficient to consider $R=(\|g\|+n\|h\|)2^nC/(1-2^n\eta)$.
  \item[\boldmath Given $n\in\N$:] The Engel property holds for all
    elements of exponent $n$ if and only if, for all $R\in\N$, the
    only cycle in $\Gamma_{n,R}$ passes through $(1,\dots,1)$.

    If the contraction coefficient satisfies $2^n\eta<1$, then it is
    sufficient to consider $R=2^nC/(1-2^n\eta)$.
  \item[\boldmath Given $G$ weakly branched and $n\in\N$:] If for some
    $R\in\N$ there exists a cycle in $\Gamma_{n,R}$ that passes
    through an element of $K^n\setminus1^n$, then no element of $G$
    whose order is a multiple of $n$ is Engel.

    If the contraction coefficient satisfies $2^n\eta<1$, then it is
    sufficient to consider $R=2^nC/(1-2^n\eta)$.
  \end{description}
\end{algorithm}

We consider the graphs $\Gamma_{n,R}$ as subgraphs of a graph
$\Gamma_{n,\infty}$ with vertex set $G^n$ and same edge definition as the
$\Gamma_{n,R}$.

We note first that, if $G$ satisfies the contraction condition
$2^\eta<1$, then all cycles of $\Gamma_{n,\infty}$ lie in fact in
$\Gamma_{n,2^nC/(1-2^n\eta)}$. Indeed, consider a cycle passing
through $(g_1,\dots,g_n)$ with $\max_i\|g_i\|=R$. Then the cycle
continues with $(g^{(1)}_1,\dots,g^{(1)}_n)$,
$(g^{(2)}_1,\dots,g^{(2)}_n)$, etc.\ with $\|g_i^{(k)}\|\le2^kR$; and
then for some $k\le n$ we have that all $g_i^{(k)}$ fix $X$; namely,
they have a trivial image in $\sym(X)$, and the map $g\mapsto g@j$ is
an injective homomorphism on them. Indeed, let
$\pi_1,\dots,\pi_n,\pi^{(i)}_1,\dots,\pi^{(i)}_n\in\Z_{/p}\subset\sym(X)$
be the images of $g_1,\dots,g_n,g^{(i)}_1,\dots,g^{(i)}_n$
respectively, and denote by $S\colon\Z_{/p}^n\to\Z_{/p}^n$ the cyclic
permutation operator. Then
$(\pi^{(n)}_1,\dots,\pi^{(n)}_n)=(S-1)^n(\pi_1,\dots,\pi_n)$, and
$(S-1)^n=\sum_j S^j\binom nj=0$ since $p|n$ and $S^n=1$. Thus there is
an edge from $(g^{(k)}_1,\dots,g^{(k)}_n)$ to
$(g^{(k+1)}_1@j,\dots,g^{(k+1)}_n@j)$ with
$\|g^{(k+1)}_i@j\|\le\eta\|g^{(k)}_i\|+C\le\eta2^nR+C$. Therefore, if
$R>2^nC/(1-2^n\eta)$ then $2^n\eta R+C<R$, and no cycle can return to
$(g_1,\dots,g_n)$.

Consider now an element $h\in G$ with $h^n=1$. For all $g\in G$, there
is an edge in $\Gamma_{n,\infty}$ from $(g,g^h,\dots,g^{h^{n-1}})$ to
$([g,h]@v,[g,h]^h@v,[g,h]^{h^{n-1}}@v)$ for some word
$v\in\{\varepsilon\}\sqcup X$, and therefore for all $c\in\N$ there
exists $d\le c$ such that, for all $v\in X^d$, there is a length-$c$
path from $(g,g^h,\dots,g^{h^{n-1}})$ to
$(E_c(g,h)@v,\dots,E_c(g,h)^{h^{n-1}}@v)$ in $\Gamma_{n,\infty}$.

We are ready to prove the first assertion: if Engel($g,h$), then
$E_c(g,h)=1$ for some $c$ large enough, so all paths of length $c$ starting
at $(g,g^h,\dots,g^{h^{n-1}})$ end at $(1,\dots,1)$. On the other hand, if
Engel($g,h$) does not hold, then all long enough paths starting at
$(g,g^h,\dots,g^{h^{n-1}})$ end at vertices in the finite graph
$\Gamma_{n,2^nC/(1-2^n\eta)}$ so must eventually reach cycles; and one of
these cycles is not $\{(1,\dots,1)\}$ since $E_c(g,h)\neq1$ for all $c$.

The second assertion immediately follows: if there exists $g\in G$
such that Engel($g,h$) does not hold, then again a non-trivial cycle
is reached starting from $(g,g^h,\dots,g^{h^{n-1}})$, and
independently of $g,h$ this cycle belongs to the graph
$\Gamma_{n,2^nC/(1-2^n\eta)}$.

For the third assertion, let
$\bar k=(k_1,\dots,k_n)\in K^n\setminus 1^n$ be a vertex of a cycle in
$\Gamma_{n,2^nC/(1-2^n\eta)}$. Consider an element $h\in G$ of order
$sn$ for some $s\in\N$. By the condition that $\#X=p$ is prime and the
image of $G$ in $\sym(X)$ is a cyclic group, $sn$ is a power of $p$,
so there exists an orbit $\{v_1,\dots,v_{sn}\}$ of $h$, so labeled
that $v_i^h=v_{i-1}$, indices being read modulo $sn$. For
$i=1,\dots,sn$ define
\[h_i=(h@v_1)^{-1}\cdots(h@v_i)^{-1},
\]
noting $h_i(h@v_i)=h_{i-1}$ for all $i=1,\dots,sn$ since
$h^{sn}=1$. Denote by `$i\%n$' the unique element of $\{1,\dots,n\}$
congruent to $i$ modulo $n$, and consider the element
\[g=\prod_{i=1}^{sn}\big(v_i*k_{i\%n}^{h_i}\big),
\]
which belongs to $G$ since $G$ is weakly branched. Let
$(k_1^{(1)},\dots,k_n^{(1)})$ be the next vertex on the cycle of
$\bar k$. We then have, using~\eqref{eq:twist},
\[[g,h]=g^{-1}g^h=\prod_{i=1}^{sn}\big(v_i*k_{i\%n}^{-h_i}\big)\prod_{i=1}^{sn}\big(v_{i-1}*k_{i\%n}^{h_i(h@v_i)}\big)=\prod_{i=1}^{sn}\big(v_i*(k_{i\%n}^{(1)})^{h_i}\big),\]
and more generally $E_c(g,h)$ and some of its states are read off the
cycle of $\bar k$. Since this cycle goes through non-trivial group
elements, $E_c(g,h)$ has a non-trivial state for all $c$, so is
non-trivial for all $c$, and Engel($g,h$) does not hold.

%%%%%%%%%%%%%%%%%%%%%%%%%%%%%%%%%%%%%%%%%%%%%%%%%%%%%%%%%%%%%%%%
\section{Proof of Theorem~\ref{thm:1}}\label{ss:proof}
The Grigorchuk group $G_0$ is contracting, with contraction
coefficient $\eta=1/2$. Therefore, the conditions of validity of
Algorithm~\ref{algo:1} are not satisfied by the Grigorchuk group, so
that it is not guaranteed that the algorithm will succeed, on a given
element $h\in G_0$, to prove that $h$ is not Engel. However, nothing
forbids us from running the algorithm with the hope that it
nevertheless terminates. It seems experimentally that the algorithm
always succeeds on elements of order $4$, and the argument proving the
third claim of Algorithm~\ref{algo:1} (repeated here for convenience)
suffices to complete the proof of Theorem~\ref{thm:1}.

Below is a self-contained proof of Theorem~\ref{thm:1}, extracting the
relevant properties of the previous section, and describing the computer
calculations as they were keyed in.

Consider first $h\in G_0$ with $h^2=1$. It follows from
Proposition~\ref{prop:torsion} that $h$ is Engel: given $g\in G_0$, we
have $E_{1+k}(g,h)=[g,h]^{(-2)^k}$ so $E_{1+k}(g,h)=1$ for $k$ larger
than the order of $[g,h]$.

For the other case, we start by a side calculation. In the Grigorchuk
group $G_0$, define $x=[a,b]$ and $K=\langle x\rangle^{G_0}$ as in
Proposition~\ref{prop:branch}, consider the quadruple
\[A_0=(A_{0,1},A_{0,2},A_{0,3},A_{0,4})=(x^{-2}x^{2ca},\,x^{-2ca}x^2x^{2cab},\,x^{-2cab}x^{-2},\,x^2)\]
of elements of $K$, and for all $n\ge0$ define
\[A_{n+1}=(A_{n,1}^{-1}A_{n,2},\,A_{n,2}^{-1}A_{n,3},\,A_{n,3}^{-1}A_{n,4},\,A_{n,4}^{-1}A_{n,1}).\]
\begin{lemma}\label{lem:calculation}
  For all $i=1,\dots,4$, the element $A_{9,i}$ fixes $111112$, is
  non-trivial, and satisfies $A_{9,i}@111112=A_{0,i}$.
\end{lemma}
\begin{proof}
  This is proven purely by a computer calculation. It is performed as
  follows within \textsc{Gap}:
\begin{Verbatim}[commandchars=\\\{\}]
gap> \textit{LoadPackage("FR");}
true
gap> \textit{AssignGeneratorVariables(GrigorchukGroup);;}
gap> \textit{x2 := Comm(a,b)^2;; x2ca := x2^(c*a);; one := a^0;;}
gap> \textit{A0 := [x2^-1*x2ca,x2ca^-1*x2*x2ca^b,(x2ca^-1)^b*x2^-1,x2];;}
gap> \textit{v := [1,1,1,1,1,2];; A := A0;; }
gap> \textit{for n in [1..9] do A := List([1..4],i->A[i]^-1*A[1+i mod 4]); od;}
gap> \textit{ForAll([1..4],i->v^A[i]=v and A[i]<>one and State(A[i],v)=A0[i]);}
true
\end{Verbatim}
\end{proof}

Consider now $h\in G_0$ with $h^2\neq1$. Again by
Proposition~\ref{prop:torsion}, we have $h^{2^e}=1$ for some minimal
$e\in\N$, which is furthermore at least $2$. We keep the notation
`$a\%b$' for the unique number in $\{1,\dots,b\}$ that is congruent to
$a$ modulo $b$.

Let $n$ be large enough so that the action of $h$ on $X^n$ has an
orbit $\{v_1,v_2,\dots,v_{2^e}\}$ of length $2^e$, numbered so that
$v_{i+1}^h=v_i$ for all $i$, indices being read modulo $2^e$. For
$i=1,\dots,2^e$ define
\[h_i=(h@v_1)^{-1}\cdots(h@v_i)^{-1},
\]
noting $h_i(h@v_i)=h_{i-1\% 2^e}$ for all $i=1,\dots,2^e$ since $h^{2^e}=1$,
and consider the element
\[g=\prod_{i=1}^{2^e}\big(v_i*A_{0,i\%4}^{h_i}\big),
\]
which is well defined since $4|2^e$ and belongs to $G_0$ by
Proposition~\ref{prop:branch}. We then have, using~\eqref{eq:twist},
\[[g,h]=g^{-1}g^h=\prod_{i=1}^{2^e}\big(v_i*A_{0,i\%4}^{-h_i}\big)\prod_{i=1}^{2^e}\big(v_{i-1\% 2^e}*A_{0,i\%4}^{h_i(h@v_i)}\big)=\prod_{i=1}^{2^e}\big(v_i*A_{1,i}^{h_i}\big),\]
and more generally
\[E_c(g,h)=\prod_{i=1}^{2^e}\big(v_i*A_{c,i}^{h_i}\big).\]
Therefore, by Lemma~\ref{lem:calculation}, for every $k\ge0$ we have
$E_{9k}(g,h)@v_0(111112)^k=A_{0,1}\neq1$, so $E_c(g,h)\neq1$ for all
$c\in\N$ and we have proven that $h$ is not an Engel element.

%%%%%%%%%%%%%%%%%%%%%%%%%%%%%%%%%%%%%%%%%%%%%%%%%%%%%%%%%%%%%%%%
\section{Other examples}
Similar calculations apply to the Gupta-Sidki group $\Gamma$. This is
another example of infinite torsion group, acting on $X^*$ for
$X=\{1,2,3\}$ and generated by the states of the following automaton:
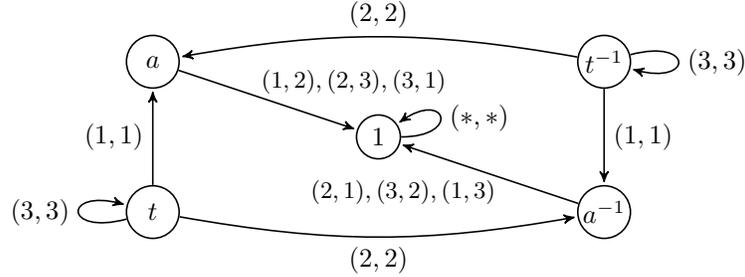
\[\begin{fsa}
  \node[state,inner sep=0pt,minimum size=7mm] (t) at (-3,-1) {$t$};
  \node[state,inner sep=0pt,minimum size=7mm] (T) at (3,1) {$t^{-1}$};
  \node[state,inner sep=0pt,minimum size=7mm] (a) at (-3,1) {$a$};
  \node[state,inner sep=0pt,minimum size=7mm] (A) at (3,-1) {$a^{-1}$};
  \node[state] (e) at (0,0) {$1$};
  \path (t) edge node {$(1,1)$} (a)
  edge [bend right=10] node[below] {$(2,2)$} (A)
  edge [loop left] node {$(3,3)$} (t)
  (T) edge node {$(1,1)$} (A)
  edge [bend right=10] node[above] {$(2,2)$} (a)
  edge [loop right] node {$(3,3)$} (T)
  (a) edge node [pos=0.2,right=5mm] {\small $(1,2),(2,3),(3,1)$} (e)
  (A) edge node [pos=0.2,left=5mm] {\small $(2,1),(3,2),(1,3)$} (e)
  (e) edge [loop,in=35,out=0,min distance=8mm,looseness=10] node[right] {$(*,*)$} (e);
\end{fsa}\]
The transformations $a,t$ may also be defined recursively by
\begin{equation}\label{eq:guptasidki}
\begin{alignedat}{3}
  (1v)^a&=2v, &\qquad (2v)^a&=3v, &\qquad (3v)^a&=1v,\\
  (1v)^t&=1v^a, &\qquad (2v)^t&=2v^{a^{-1}}, &\qquad (3v)^t&=3v^t.
\end{alignedat}
\end{equation}
The corresponding result is
\begin{theorem}\label{thm:gs}
  The only Engel element in the Gupta-Sidki group $\Gamma$ is the identity.
\end{theorem}
We only sketch the proof, since it follows that of Theorem~\ref{thm:1}
quite closely. Analogues of Propositions~\ref{prop:torsion}
and~\ref{prop:branch} hold, with $[\Gamma,\Gamma]$ in the r\^ole of
$K$. An analogue of Lemma~\ref{lem:calculation} holds with
$A_0=([a^{-1},t],[a,t]^a,[t^{-1},a^{-1}])$ and $A_{4,i}@122=A_{0,i}$.

%%%%%%%%%%%%%%%%%%%%%%%%%%%%%%%%%%%%%%%%%%%%%%%%%%%%%%%%%%%%%%%%
\section{Closing remarks}
It would be dishonest to withhold from the reader how I arrived at the
examples given for the Grigorchuk and Gupta-Sidki groups. I started by
small words $g,h$ in the generators of $G_0$, respectively $\Gamma$,
and computed $E_c(g,h)$ for the first few values of $c$. These
elements are represented, internally to \textsc{Fr}, as Mealy
automata. A natural measure of the complexity of a group element is
the size of the minimized automaton, which serves as a canonical
representation of the element.

For some choices of $g,h$ the size increases exponentially with $c$,
limiting the practicality of computer experiments. For others (such as
$(g,h)=((ba)^4c,ad)$ for the Grigorchuk group), the size increases
roughly linearly with $c$, making calculations possible for $c$ in the
hundreds. Using these data, I guessed the period $p$ of the recursion
($9$ in the case of the Grigorchuk group), and searched among the
states of $E_c(g,h)$ and $E_{c+p}(g,h)$ for common elements; in the
example, I found such common states for $c=23$. I then took the
smallest-size quadruple of states that appeared both in $E_c(g,h)$ and
$E_{c+p}(g,h)$ and belonged to $K$, and expressed the calculation
taking $E_c(g,h)$ to $E_{c+p}(g,h)$ in the form of
Lemma~\ref{lem:calculation}.

It was already shown by Bludov~\cite{bludov:engel} that the wreath
product $G_0^4\rtimes D_4$ is not Engel. He gave, in this manner, an
example of a torsion group in which a product of Engel elements is not
Engel. Our proof is a refinement of his argument.

A direct search for elements $A_{0,1},\dots,A_{0,4}$ would probably
not be successful, and has not yielded simpler elements than those
given before Lemma~\ref{lem:calculation}, if one restricts them to
belong to $K$; one can only wonder how Bludov found the quadruple
$(1,d,ca,ab)$, presumably without the help of a computer.

%%%%%%%%%%%%%%%%%%%%%%%%%%%%%%%%%%%%%%%%%%%%%%%%%%%%%%%%%%%%%%%%
\section*{Acknowledgments}
I am grateful to Anna Erschler for stimulating my interest in this question
and for having suggested a computer approach to the problem. %, and to Ines
%Klimann and Matthieu Picantin for helpful discussions that have improved the
%presentation of this note.

%%%%%%%%%%%%%%%%%%%%%%%%%%%%%%%%%%%%%%%%%%%%%%%%%%%%%%%%%%%%%%%%

\end{document}